\documentclass[runningheads]{llncs}
\usepackage{version}
\includeversion{LONG}
\excludeversion{SHORT} 

\usepackage{xspace}
\usepackage{url}

\usepackage{etex}
\usepackage{amsmath}
\usepackage[english]{babel}
\usepackage{amsfonts}
\usepackage{amssymb}
\usepackage{euscript}
\usepackage{enumerate}
\usepackage[dvips]{graphicx}
\usepackage{rotating}
\usepackage{proof}
\usepackage{bussproofs}
   \EnableBpAbbreviations
   \def\extraVskip{0.4em}
\usepackage[dvips,ps,all,rotate]{xy}
\usepackage{pdfsync}

\newcommand{\tr}[1]{#1^*}

\newcommand{\interp}{\mathcal{I}}

\newcommand{\ltl}{\mathit{LTL}}
\newcommand{\Gammaltl}{\Lambda}
\newcommand{\Gammaltlh}{\Gamma}
\newcommand{\ltlh}{\mathit{LTL_{\history}}}

\newcommand{\ctlstar}{\mathit{CTL^*}}

\newcommand{\axltl}{\mathcal{H}(\ltl)}

\newcommand{\nltlh}{\mathcal{N}(\ltlh)}

\newcommand{\m}{\mathcal{M}}

\newcommand{\n}{\mathcal{N}}

\newcommand{\s}{\mathcal{S}}

\newcommand{\lab}{\mathcal{L}}
\newcommand{\val}{\mathcal{V}}
\newcommand{\prop}{\mathcal{P}}

\newcommand{\bottomE}{\bottom\!\!E}

\newcommand{\limpliesI}{\limplies\!\!I}
\newcommand{\limpliesE}{\limplies\!\!E}
\newcommand{\wedgeI}{\wedge I}
\newcommand{\wedgeE}{\wedge E}
\newcommand{\veeI}{\vee I}
\newcommand{\veeE}{\vee E}

\newcommand{\historyI}{\history I}
\newcommand{\historyE}{\history E}

\newcommand{\lin}{\mathit{lin}_\nextb}

\newcommand{\ser}{\mathit{ser}_\nextb}

\newcommand{\refldesc}{\mathit{refl}_\descb}

\newcommand{\transdesc}{\mathit{trans}_\descb}
\newcommand{\basedesc}{\mathit{base}_\descb}
\newcommand{\eqdesc}{\mathit{eq}_\descb}
\newcommand{\splitdesc}{\mathit{split}_\descb}
\newcommand{\last}{\mathit{last}}

\newcommand{\ind}{\mathit{ind}}

\newcommand{\base}{\mathit{base}_\descb}
\newcommand{\gI}{\g I}
\newcommand{\gE}{\g E}
\newcommand{\xI}{\x I}
\newcommand{\xE}{\x E}
\newcommand{\fI}{\f I}
\newcommand{\fE}{\f E}
\newcommand{\bottom}{\perp}

\newcommand{\nextb}{\lhd}
\newcommand{\descb}{\leqslant}

\newcommand{\nmodels}{\nvDash}
\newcommand{\modelsltl}{\models_{_\ltl}}
\newcommand{\nmodelsltl}{\nmodels_{_\ltl}}
\newcommand{\modelsltlh}{\models_{_\history}}
\newcommand{\nmodelsltlh}{\nmodels_{_\history}}
\newcommand{\vdashltlh}{\vdash_{_\history}}
\newcommand{\vdashltl}{\vdash_{_\ltl}}

\newcommand{\limplies}{\supset}

\newcommand{\g}{\mathsf{G}}
\newcommand{\history}{\nabla}

\newcommand{\nat}{\mathbb{N}}

\newcommand{\f}{\mathsf{F}}

\newcommand{\x}{\mathsf{X}}
\newcommand{\until}{\mathsf{U}}

\newcommand{\ltllocal}{\ltl^l}
\newcommand{\ltlhistory}{\ltl^{l\history}}
\newcommand{\local}[1]{{#1}^l}
\newcommand{\hist}[1]{{#1}^{l\history}}

\newcommand{\secref}[1]{Section~\ref{#1}}

\input xy
\xyoption{all}

\begin{document}


\title{A History of Until}

\author{Andrea Masini \qquad Luca Vigan{\`o} \qquad Marco Volpe}

\institute{Department of Computer Science, University of Verona, Italy \\
    \{\url{andrea.masini} $\mid$ \url{luca.vigano} $\mid$ \url{marco.volpe}\}\url{@univr.it} }



%
 \maketitle

\begin{abstract}
Until is a notoriously difficult temporal operator as it is both
existential and universal at the same time: $A \until B$ holds at the
current time instant $w$ iff either $B$ holds at $w$ or there
\emph{exists} a time instant $w'$ in the future at which $B$ holds and such
that $A$ holds in \emph{all} the time instants between the current one
and $w'$.
This ``ambivalent'' nature poses a significant challenge when attempting
to give deduction rules for until. In this paper, in contrast,
we make explicit this duality of until 
by introducing a new temporal operator $\history$ that allows us to
formalize the ``history'' of until, i.e.,~the ``internal'' universal
quantification over the time instants between the current one and $w'$.
This approach provides the basis for formalizing deduction systems for
temporal logics endowed with the until operator. For concreteness, we
give here a labeled natural deduction system for a linear-time logic
endowed with the new history operator and show that, via a proper
translation, such a system is also sound and complete with respect to
the linear temporal logic $\ltl$ with until.
%
\end{abstract}
%

\section{Introduction}\label{sec:introduction}

Until is a notoriously difficult temporal operator. This is because of
its ``ambivalent'' nature of being an operator that is both existential
and universal at the same time: $A \until B$ holds at the current time
instant (sometimes ``world'' or ``state'' is used in place of ``time
instant'') $w$ iff either $B$ holds at $w$ or there \emph{exists} a time
instant $w'$ in the future at which $B$ holds and such that $A$ holds in
\emph{all} the time instants between the current one and $w'$. The words
in emphasis highlight the dual existential and universal nature of $\until$, which
poses a significant challenge when attempting to give deduction rules
for until, so that deduction systems for temporal logics  either
deliberately exclude until from the set of operators considered or
devise clever ways to formalize reasoning about until. And even if one
manages to give rules, these often come at the price of additional
difficulties for, or even the impossibility of, proving useful
metatheoretic properties, such as normalization or the subformula
property. (This is even more so in the case of
Hilbert-style axiomatizations, which provide axioms for until, but are
not easily usable for proof construction.) See, for instance,
\cite{BCRV09-JLC,Bol+:AutomatedLTL:07,FisherGabbayVila05,gore:99,Gough84,Schwendimann98},
where techniques for formalizing suitable inference rules include
introducing additional information (such as the use of a Skolem function
$f(A \until B)$ to name the time instant where $B$ begins to hold), or
exploiting the standard recursive unfolding of until
\begin{equation}\label{eq:unfolding} 
	A \until B \equiv B \vee (A \wedge \x(A \until B))
\end{equation}
which says that $A \until B$ iff either $B$ holds or $A$ holds and in the successor
time instant (as expressed by the \emph{next} operator $\x$) we have
again $A \until B$.

In this paper, in contrast, we make explicit the duality of until
by introducing a new temporal operator $\history$ that allows us to
formalize the ``history'' of until, i.e.,~the fact that when we have $A
\until B$ the formula $A$ holds in all the time instants between the
current one and the one where $B$ holds. We express this ``historic''
universal quantification by means of 
$\history$ with respect to the following intuitive translation:
\begin{equation}\label{eq:translation}
A \until B \equiv B \vee \f(\x B \wedge \history A)
\end{equation}
That is: $A \until B$ iff either $B$ holds or there exists a time
instant $w'$ in the future (as expressed by the \emph{sometime in the
future} operator $\f$) such that
\begin{itemize}
	\item $B$ holds in the successor time instant, and 
	\item $A$ holds in all the time instants between the current one and $w'$ (included).
\end{itemize}
The latter conjunct is precisely what the \emph{history} operator
$\history$ expresses\footnote{This is in contrast to the
unfolding \eqref{eq:unfolding}. The decoupling of $\until$ that we
achieve with $\history$ is precisely what allows us to give well-behaved
(in a sense made clearer below) natural deduction rules.}. This is better
seen when introducing labeling: since $\history$ actually quantifies
over the time instants in an interval (delimited by the current 
instant and the one where the $B$ of the until holds), 
we adopt a labeling discipline that is slightly different from
the more customary one of labeled deduction.

The framework of \emph{labeled deduction} has been successfully employed
for several non-classical, and in particular modal and temporal, logics,
e.g.,~\cite{Gab:LabDedSys:96,Sim:PhdThesis:94,Vig:Labelled:00}, since
labeling provides a clean and effective way of dealing with modalities
and gives rise to deduction systems with good proof-theoretical
properties. The basic idea is that labels allow one to explicitly encode
additional information, of a semantic or proof-theoretical nature, that
is otherwise implicit in the logic one wants to capture. So, for
instance, instead of a formula $A$, one can consider the \emph{labeled
formula} $b:A$, which intuitively means that $A$ holds at the time
instant denoted by $b$ within the underlying Kripke semantics. One can
also use labels to specify how time instants are related, e.g.,~the
\emph{relational formula} $b R c$ states that the time instant $c$ is
accessible from $b$.

Considering labels that consist of a single time instant is not enough
for $\history$, as the operator is explicitly designed to speak about a
sequence of time instants (namely, the ones constituting the history of
the corresponding until, if indeed $\history$ results from the
translation of an $\until$). We thus consider labels that are built out
of a sequence of time instants, so that we can write $\alpha b_1 b_3:
\history A$ to express, intuitively, that $A$ holds in the interval
between time instants $b_1$ and $b_3$, which together with the
sub-sequence $\alpha$ constitute a sequence of time instants $\alpha b_1
b_3$. This allows us to give the natural deduction elimination
rule 
\begin{displaymath}\small
\infer[\historyE]{\alpha b_1 b_2: A}{\alpha b_1 b_3: \history A & b_1 \descb b_2 & b_2 \descb b_3}
\end{displaymath}
that says that if $\history A$ holds at time instant $b_3$ at the end
of the sequence $\alpha b_1 b_3$ and if 
$b_2$ is in-between $b_1$ and $b_3$, as expressed by the relational
formulas with the accessibility relation $\descb$, then we can conclude
that $A$ holds at $b_2$.

Dually, we can introduce $\history A$ at time instant $b_3$ at the end
of the sequence $\alpha b_1 b_3$ whenever from the assumptions $b_1
\descb b_2$ and $b_2\descb b_3$ for a \emph{fresh} $b_2$ we can infer
$\alpha b_1 b_2:A$, i.e.\footnote{The side condition that $b_2$ is
fresh means that $b_2$ is different from $b_1$ and $b_3$, and
does not occur in any assumption on which $\alpha b_1 b_2:A$ depends
other than the discarded assumptions $b_1 \descb b_2$ and $b_2 \descb
b_3$.},:
\begin{displaymath}\small
	\infer[\historyI]{\alpha b_1 b_3: \history A}{\infer*{\alpha b_1 b_2:A}{[b_1 \descb b_2]\, [b_2\descb b_3]}}
\end{displaymath}

The adoption of time instant sequences for labels has thus allowed us to
give rules for $\history$ that are well-behaved in the spirit of natural
deduction~\cite{Pra:NatDed:65}: there is precisely one introduction and
one elimination rule for $\history$, as well as for the other
connectives and temporal operators ($\limplies$, $\g$, and $\x$). This
paves the way to a proof-theoretical analysis of the resulting natural
deduction systems, e.g.,~to show proof normalization and other useful
meta-theoretical analysis, which we are tackling in current work.

Moreover, the rules $\historyI$ and $\historyE$ provide a clean-cut way
of reasoning about until, according to the translation
\eqref{eq:translation}, provided that we also give rules for $\f$ and
$\x$. These operators have a local nature, in the sense that they speak
not about sequences of time instants but about single time instants.
Still, we can easily give natural deduction rules for them by
generalizing the more standard ``single-time instant'' rules
(e.g.,~\cite{BCRV09-JLC,Bol+:AutomatedLTL:07,gore:99,MasiniViganoVolpe09a,Sim:PhdThesis:94,Vig:Labelled:00,ViganoVolpe08})
using our labeling with sequences of time instants. As we will discuss
in more detail below, if we collapse the sequences of time instants to
consider only the final time instant in the sequence (or, equivalently,
if we simply ignore all the instants in a sequence but the last), then
these rules reduce to the standard ones. For instance, for the
\emph{always in the future} operator $\g$ (the 
dual of $\f$)
and $\x$, with the corresponding \emph{successor relation} $\nextb$, we
can give the elimination rules
\begin{displaymath}\small
	\infer[\gE]{\alpha b_1 b_2: A}{\alpha b_1: \g A & b_1 \descb b_2}
	\qquad \text{and} \qquad 
	\infer[\xE]{\alpha b_1b_2: A}{\alpha b_1: \x A & b_1 \nextb b_2}
\end{displaymath}
The rule $\gE$ says that if $\g A$ holds at time instant $b_1$, which is
the last in the sequence $\alpha b_1$ and $b_2$ is $\descb$-accessible
from $b_1$ (i.e.,~$b_1 \descb b_2$), then we can conclude that $A$ holds
for the sequence $\alpha b_1 b_2$. The rule $\xE$ is justified similarly
(via $\nextb$).
The corresponding introduction rules are given in \secref{sec:system},
together with rules for $\bottom$ and the connective $\limplies$, as
well as a rule for induction on the underlying linear ordering. As we
will see, we also need rules expressing the properties of the relations
$\descb$ and $\nextb$. Moreover, the fact that we consider sequences of
time instants as labels requires us to consider some structural rules to
express properties of such sequences (with respect to formulas). 

This approach thus provides the basis for formalizing deduction systems
for temporal logics endowed with the until operator. For concreteness,
we give here a labeled natural deduction system for a linear-time logic
endowed with the new history operator $\history$ and show that, via a
proper translation, such a system is also sound and complete with
respect to the linear temporal logic $\ltl$ with until. (We do not
consider past explicitly here, but adding operators and rules for it
should be unproblematic, e.g.,~as in~\cite{ViganoVolpe08}.)


We proceed as follows. In \secref{sec:ltl}, we
briefly recall the syntax and semantics of $\ltl$, and an axiomatization
for it. In \secref{sec:ltlh}, we define $\ltlh$, the logic that is
obtained from $\ltl$ by replacing $\until$ with the history $\history$,
and give a validity-preserving translation, based on
\eqref{eq:translation}, from $\ltl$ into $\ltlh$. In
\secref{sec:system}, we give a labeled natural deduction system $\nltlh$ 
that it is sound with respect to the semantics of $\ltlh$. 
By focusing only on those derivations
whose conclusion and open assumptions correspond to the
translation of $\ltl$-formulas, we show that $\nltlh$ can be used to capture reasoning in $\ltl$ and that it is in fact sound and complete
with respect to the semantics of $\ltl$. In \secref{sec:discussion}, we draw conclusions and illustrate directions of current and
future work.
\begin{SHORT}
Due to lack of space, some of the proofs are only sketched here. Full proofs are given in \cite{}.
\end{SHORT}
\begin{LONG}
Full proofs are given in the appendix.
\end{LONG}

\section{The Linear Temporal Logic $\ltl$}\label{sec:ltl}

We  recall the syntax and
semantics of $\ltl$ and an axiomatization for it.


\begin{definition}\em
Given a set $\cal{P}$ of propositional symbols, the set of
\emph{(well-formed) $\ltl$-formulas} 
is defined by the grammar
\begin{displaymath}
  A ::= \, p \mid \bottom \mid A \limplies A \mid \g A \mid \x A \mid A \until A 
\end{displaymath}
where $p \in \cal{P}$.
The set of $\ltl$-\emph{atomic formulas} is $\cal{P}\cup\{\bot\}$. The
\emph{complexity} of an $\ltl$-formula is the number of occurrences of
the connective $\limplies$ and of the temporal operators $\g$, $\x$, and
$\until$.
\end{definition}

The intuitive meaning of $\g$, $\x$, and $\until$ is the standard one:
$\g A$ states that $A$ holds always in the future, 
$\x A$ states that $A$ holds in the next time instant, and
$A \until B$ states that $B$ holds at the current time instant or there is a time instant $w$ in the future such that $B$ holds in $w$ and $A$ holds in all the time instants between the current one and $w$.
As usual, we can introduce abbreviations and use, e.g., $\neg$, $\vee$
and $\wedge$ for negation, disjunction, and conjunction, respectively:
$\neg A \equiv A\limplies \bottom$, $A \vee B \equiv \neg A \limplies
B$, and $A \wedge B \equiv \neg (\neg A \vee \neg B)$. We can also
define other temporal operators, e.g.,~$\f A \equiv \neg \g \neg A$ to
express that $A$ holds sometime in the future. We write $\Gammaltl$
to denote a \emph{set of $\ltl$-formulas}.



\begin{definition}\label{def:ltl-truth}\em
Let $\n=\langle \nat, s:\nat \to\nat, \leq \rangle$ be the
standard structure of natural numbers, where $s$ and $\leq$ are
respectively the successor function and the total (reflexive) order
relation. An \emph{$\ltl$-model} is a pair $\m=\langle \n, \val
\rangle$ where $\val:\nat \to 2^{\prop}$.
\emph{Truth} for an $\ltl$-formula at a point $n \in \nat$ in an
$\ltl$-model $\m = \langle \n, \val\rangle$ is the
smallest relation $\modelsltl$ satisfying:
\begin{SHORT}
\vspace*{-0.45cm}
\end{SHORT}
	\begin{eqnarray*}
		\m, n \modelsltl p 
		& \quad \text{iff} \quad & p \in \val(n) \\
		\m, n \modelsltl A \limplies B 
		& \quad \text{iff} \quad & \m, n \modelsltl A 
		\text{ implies } \m, n \modelsltl B \\
		\m, n \modelsltl \g A 
		& \quad \text{iff} \quad & \m, m \modelsltl A 
		\text{ for all } m \ge n\\
		\m, n \modelsltl \x A 
		& \quad \text{iff} \quad & \m, n+1 \modelsltl A \\
		\m, n \modelsltl A \until B 
		& \quad \text{iff} \quad & \text{there exists } n' \ge n \text{ such that } \m,n'\modelsltl B\\ && \text{and } \m,m \modelsltl A \text{ for all } n \le m < n'
	\end{eqnarray*}
Note that $\m,n \nmodelsltl\ \bottom$ for every $\m$ and
$n$. 
By extension, we write:
\begin{SHORT}
\vspace*{-0.45cm}
\end{SHORT}
\begin{eqnarray*}
  \m \modelsltl A & \quad \text{iff} \quad & \m,n \modelsltl A \mbox{ for every natural number } n\\
  \m \modelsltl \Gammaltl & \quad \text{iff} \quad & \m \modelsltl A \mbox{ for all } A \in \Gammaltl\\
  \Gammaltl \modelsltl A & \quad \text{iff} \quad & \m \modelsltl \Gammaltl \mbox{ implies } \m \modelsltl A\mbox{, for every } \ltl\mbox{-model } \m
	\end{eqnarray*}
\end{definition}



We now present a sound and complete Hilbert-style axiomatization, which we call $\axltl$, for $\ltl$ (see, e.g., \cite{Gol:LogicsTimeComp:87}). 
$\axltl$ consists of the axioms
\begin{SHORT}
\vspace*{-0.15cm}
\end{SHORT}
	\begin{displaymath}\small
		\renewcommand{\arraystretch}{1.2}
		\begin{array}{lclllcl}
		\mathit{(A1)} && \textrm{Any tautology instance}
		&&
		\mathit{(A2)} && \g(A \limplies B) \limplies (\g A \limplies \g B)
		 \\
		\mathit{(A3)} && (\x \neg A \leftrightarrow \neg \x A)
		&&
		\mathit{(A4)} && \x(A \limplies B) \limplies (\x A \limplies \x B)
		 \\
		\mathit{(A5)} && \g A \, \limplies \, A  \wedge \x \g A
		&&
		\mathit{(A6)} && \g (A \limplies \x A) \limplies (A \limplies \g A) \\
		\mathit{(A7)} && A \until B \, \leftrightarrow \, (B \vee (A \wedge \x (A \until B)))
		&&
		\mathit{(A8)} && A \until B \, \limplies \, \f B
		\end{array}
	\end{displaymath}	
where we denote with $\leftrightarrow$ the double implication, and of
the rules of inference
\begin{SHORT}
\vspace*{-0.15cm}
\end{SHORT}
\begin{displaymath}\small
	\renewcommand{\arraystretch}{1.2}
		\begin{array}{l}
		\mathit{(MP)} \ \textrm{If } A \textrm{ and } A \limplies B \textrm{ then } B 
		\qquad
		\mathit{(Nec_X)} \ \textrm{If } A \textrm{ then } \x A 
		\qquad
		\mathit{(Nec_G)} \ \textrm{If } A \textrm{ then } \g A 
		\end{array}
\end{displaymath}

The set of theorems of $\axltl$ is the smallest set containing these
axioms and closed with respect to these rules of inference.

\section{$\ltlh$: $\ltl$ with history}\label{sec:ltlh}

In this section, we give the linear temporal logic $\ltlh$, which is
obtained from $\ltl$ by replacing the operator $\until$ with a new unary
temporal operator $\history$, called \emph{history}. The definition of
the semantics of $\ltlh$ requires a notion of truth given with respect
to sequences of time instants rather than just to time instants. We will
then provide a translation from the language of $\ltl$ into the language
of $\ltlh$ and show some properties of such a translation.

\subsection{Syntax and semantics}\label{sub:ltlh-semantics}

\begin{definition}\em
Given a set $\cal{P}$ of propositional symbols, the set of
\emph{(well-formed) $\ltlh$-formulas} 
is defined by the grammar
\begin{displaymath}
  A ::= \, p \mid \bottom \mid A \limplies A \mid \g A \mid \x A \mid \history A
\end{displaymath}
where $p \in \cal{P}$.
The set of $\ltlh$-\emph{atomic formulas} is $\cal{P}\cup\{\bot\}$. The
\emph{complexity} of an $\ltlh$-formula is the number of occurrences of
the connective $\limplies$ and of the temporal operators $\x$, $\g$, and
$\history$.
\end{definition}


The intuitive meaning of the operators $\g$ and $\x$ is the same as for
$\ltl$, while $\history A$ intuitively states that $A$ holds at any instant of a
 particular time interval (but here we see that we need sequences of time
instants to formalize the semantics of the history operator, as we
anticipated in the introduction). Again, we can define other connectives
and operators as abbreviations, e.g.,~$\neg$, $\vee$, $\wedge$, $\f$ and
so on. We write $\Gammaltlh$ to denote a \emph{set of $\ltlh$-formulas}.



To define a labeled deduction system for the logic $\ltlh$, we extend
the language with a set of labels and finite sequences of labels, and
introduce the notions of labeled formula and relational formula.

\begin{definition}\em

Let $\lab$ be a set of labels. A finite non-empty sequence of labels
(namely, an element of $\lab^+$) is called a \emph{sequence}. 
If $A$ is an $\ltlh$-formula and $\alpha$ is a sequence, then
$\alpha:A$ is a \emph{labeled (well-formed) formula} (\emph{lwff} for
short). 
The set of \emph{relational (well-formed) formulas} (\emph{rwffs} for
short) is the set of expressions of the form $b \descb c$ or $b \nextb
c$, where $b$ and $c$ are labels.
\end{definition}

In the rest of the paper, we will assume given a fixed denumerable set
$\lab$ of labels and we will use 
$b, c, d, \ldots$ 
to denote labels, 
$\alpha,\beta, \gamma$ to denote finite 
sequences of labels\footnote{With a slight abuse of notation, we will also use $\alpha,\beta, \gamma$ to denote possibly empty subsequences and thus write $\alpha b_1 \ldots b_k$ (for $k \ge 1$) to denote a sequence where $\alpha$ may be empty.}
(e.g.,~$bcd\ldots$ or just $b$ in the case
of a sequence consisting of only one time instant),
$\varphi$ to denote a \emph{generic formula} (either labeled or relational) and $\Phi$ to denote a \emph{set of generic formulas}.

%
%


\begin{definition}\label{def:ltlh-truth}\em

An \emph{observation sequence} is a non-empty sequence $ \sigma= [n_0,$
$\ldots,$ $n_k]$ of natural numbers.
\emph{Truth} for an $\ltlh$-formula at an observation sequence $\sigma$ in an
$\ltl$-model $\m = \langle \n, \val\rangle$ is the
smallest relation $\modelsltlh$ satisfying:
\begin{SHORT}
\vspace*{-0.5cm}
\end{SHORT}
		\begin{eqnarray*}
		\m, [n_0,\ldots,n_k] \modelsltlh p 
		& \quad \text{iff} \quad & p \in \val(n_k) \\
		\m,[n_0,\ldots,n_k] \modelsltlh A \limplies B
		& \quad \text{iff} \quad & \m,[n_0,\ldots,n_k] \modelsltlh A 
		\text{ implies }\\
 & & \m,[n_0,\ldots,n_k] \modelsltlh B \\
		\m,[n_0,\ldots,n_k]\modelsltlh \g A 
		& \quad \text{iff} \quad & \m,[n_0,\ldots,n_k,m] \modelsltlh A 
		\text{ for all } m \ge n_k\\
		\m,[n_0,\ldots,n_k]\modelsltlh \x A 
		& \quad \text{iff} \quad & \m,[n_0,\ldots,n_k,n_k+1]\modelsltlh A \\
		\m,[n_0,\ldots,n_{k-1},n_k] \modelsltlh \history A 
		& \quad \text{iff} \quad & \m,[n_0,\ldots,n_{k-1},m]\modelsltlh A\\
& \quad & \text{for all } n_{k-1} \le m \le n_k \text{ (if } 0 < k\text{)}
\end{eqnarray*}
\begin{eqnarray*}
		\m,[n_0]\modelsltlh \history A 
		& \quad \text{iff} \quad & \m,[n_0]\modelsltlh A		
		\end{eqnarray*}
By extension, we write:
\begin{SHORT}
\vspace*{-0.35cm}
\end{SHORT}
	\begin{eqnarray*}
  \m \modelsltlh A & \quad \mbox{iff} \quad & \m,\sigma \modelsltlh A \mbox{ for every observation sequence } \sigma\\
  \m \modelsltlh \Gammaltlh & \quad \mbox{iff} \quad & \m \modelsltlh A \mbox{ for all } A \in \Gammaltlh\\
  \Gammaltlh \modelsltlh A & \quad \mbox{iff} \quad & \m \modelsltlh \Gammaltlh \mbox{ implies } \m \modelsltlh A\mbox{, for every } \ltl\mbox{-model } \m
\end{eqnarray*}
Given an $\ltl$-model $\m$, a \emph{structure} is a pair $\s=\langle \m,
\interp \rangle$ where $\interp:\lab\to\nat$. Let $\Sigma$ be the set of observation sequences and $\interp^+:\lab^+ \to \Sigma$ the extension of $\interp$ to sequences,
i.e.,~$\interp^+(b_0\ldots b_n)$ $=[\interp(b_0),\ldots, \interp(b_n)]$.
\emph{Truth} for a generic formula $\varphi$ in a structure
$\s=\langle\m, \interp\rangle$ is the smallest
relation $\modelsltlh$ satisfying:
\begin{SHORT}
\vspace*{-0.35cm}
\end{SHORT}
		\begin{eqnarray*}
		\m,\interp \modelsltlh a \descb b
		& \quad \text{iff} \quad & \interp(a) \le \interp(b) \\
	    \m,\interp \modelsltlh a \nextb b
		& \quad \text{iff} \quad & \interp(b) = \interp(a) + 1 \\
	    \m,\interp \modelsltlh \alpha:A
		& \quad \text{iff} \quad & \m, \interp^+(\alpha)\modelsltlh A
		\end{eqnarray*}
Note that $\m,\sigma \nmodelsltlh \bot$ and $\m,\interp\ \nmodelsltlh \alpha:\bot$ for every $\m$, $\sigma$ and
$\interp$.

Given a set $\Phi$ of generic formulas and a generic formula $\varphi$:
\begin{SHORT}
\vspace*{-0.35cm}
\end{SHORT}
		\begin{eqnarray*}
		\m,\interp \modelsltlh \Phi 
		& \quad \text{iff} \quad & \m,\interp \modelsltlh \varphi 
		\text{ for all } \varphi\in\Phi \\
		\Phi \modelsltlh \varphi 
		& \quad \text{iff} \quad & \m, \interp \modelsltlh \Phi 	
		\text{ implies } \m, \interp \modelsltlh \varphi 
		\text{ for all } \m \text{ and } \interp
		\end{eqnarray*}
\end{definition}

\subsection{A translation from $\ltl$ into $\ltlh$}\label{sub:translation}

$\ltl$ and $\ltlh$ are, obviously, related logics. 
In fact, below we will define a validity-preserving translation
$\tr{(\cdot)}$ from $\ltl$ into $\ltlh$. Then, in Lemma~\ref{lem:last},
we will show that if an $\ltlh$-formula corresponds to the translation
of some $\ltl$-formula, then it can be interpreted ``locally'', i.e.,~its
truth value with respect to an observation sequence depends only on the
last element of the sequence. Finally, in Lemma~\ref{lem:translation} and
Theorem~\ref{th:translation}, we will use this result to prove that the
translation preserves the validity of formulas. 
This property allows us to 
use the deduction system for $\ltlh$, which will be presented in
Section~\ref{sec:system}, for reasoning on $\ltl$ too, as we will show
in Section~\ref{sub:sandc}, when discussing soundness and completeness of
the system.

\begin{definition}\em
	We define the translation $\tr{(\cdot)}$ from the language of $\ltl$ into the language of $\ltlh$ inductively as follows: \\
	   \begin{center}
      \begin{tabular}{lcllll}
        $\tr{(p)}$ & $=$ & $p \, ,$ for $p$ atomic
		\\ 
        $\tr{(\g A)}$ & $=$ & $\g \, \tr{(A)} \,$\\
        $\tr{(\bottom)}$ & $=$ & $\bottom \,$
		\\
        $\tr{(\x A)}$ & $=$ & $\x \, \tr{(A)} \,$\\
        $\tr{(A \limplies B)}$ & $=$ & $\tr{(A)} \limplies \tr{(B)} \,$
		\\
        $\tr{(A \until B)}$ & $=$ & $\tr{(B)} \vee (\, \f \, (\, \x \, \tr{(B)} \wedge \history \tr{(A)}\, )) \,$
      \end{tabular}
\end{center}

We extend $\tr{(\cdot)}$ to sets of formulas in the obvious way:
$\tr{\Gammaltl} = \{ \tr{B} \, \mid \, B \in \Gammaltl \}$.
\end{definition}


\begin{lemma}
	\label{lem:last}
	Let $\m$ be an $\ltl$-model, $[n_1, \ldots, n_k]$ an observation sequence, and $A$ an $\ltl$-formula. Then $\m, [n_1, \ldots, n_k] \modelsltlh \tr{A} \;$ iff $\m, [m_1, \ldots, m_r, n_k]$ $\modelsltlh \tr{A}$ for every sequence $m_1, \ldots, m_r$.  
\end{lemma}
\begin{corollary}
	\label{cor:last}
	Let $\m$ be an $\ltl$-model, $[n_1, \ldots, n_k]$ an observation sequence, and $A$ an $\ltl$-formula. Then $\m,[n_1, \ldots, n_k] \modelsltlh \tr{A}$ iff $\m,[n_k] \modelsltlh \tr{A}$.
\end{corollary}
\begin{lemma}
	\label{lem:translation}
	Let $\m$ be an $\ltl$-model, $n$ a natural number, and $A$ an $\ltl$-formula. Then $\m, n \modelsltl A$ iff $\m, [n] \modelsltlh \tr{A}$.
\end{lemma}
\begin{theorem}
	\label{th:translation}
	Let $\Gammaltl$ be a set of $\ltl$-formulas, $A$ an $\ltl$-formula and $\tr{\Gammaltl} = \{ \tr{B} \, \mid \, B\in \Gammaltl \}$. Then $\,\Gammaltl \modelsltl A\,$ iff $\,\tr{\Gammaltl} \modelsltlh \tr{A} \, .$
\end{theorem}
\begin{proof}
  By Definition \ref{def:ltl-truth}, $\Gammaltl \modelsltl A\,$ iff $\,\forall \m.\, \m\modelsltl \Gammaltl$ implies $\m \modelsltl A \,$ iff $\,\forall \m.\, (\, \forall B\in\Gammaltl.\, \forall n.\, \m,n \modelsltl B$ implies $\forall n. \, \m,n \modelsltl A \,)\,$ iff (by Lemma \ref{lem:translation}) $\,\forall \m.\, (\, \forall B\in\Gammaltl.\, \forall n.\, \m,[n] \modelsltlh \tr{B}$ implies $\forall n. \, \m,[n] \modelsltlh \tr{A} \,)\,$ iff (by Lemma \ref{lem:last}) $\,\forall \m.\, (\, \forall B\in\Gammaltl.\, \forall \sigma.\, \m,\sigma \modelsltlh \tr{B}$ implies $\forall \sigma. \, \m,\sigma \modelsltlh \tr{A} \,)\,$
iff (by Definition \ref{def:ltlh-truth}) $\,\forall \m.\, (\, \forall B\in\Gammaltl.\, \m \modelsltlh \tr{B}$ implies $\m \modelsltlh \tr{A} \,)\,$
iff $\,\forall \m.\, (\, \m \modelsltlh \tr{\Gammaltl}$ implies $\m \modelsltlh \tr{A} \,)\,$
iff $\,\tr{\Gammaltl} \modelsltlh \tr{A}$.
\end{proof}


\section{$\nltlh$: a labeled natural deduction system for $\ltlh$}\label{sec:system}

In this section, we will first define a labeled natural deduction system
$\nltlh$ that is sound with respect to the semantics of $\ltlh$. Then,
by considering a restriction of the set of $\nltlh$-derivations and by
using the translation $\tr{(\cdot)}$ and the related results, we will
show that $\nltlh$ can be also used for reasoning on $\ltl$: we will
prove soundness with respect to the semantics of $\ltl$ and we will give
a proof of weak completeness with respect to $\ltl$, by exploiting the
Hilbert-style axiomatization $\axltl$.


\subsection{The rules of $\nltlh$}\label{sub:rules}

\begin{figure}[tp]
  \scriptsize
  \begin{displaymath}
	\hspace*{-1cm}
    \renewcommand{\arraystretch}{4}
    \begin{array}{c}
      \infer[\bottomE]{\alpha_1: A}{\infer*{\alpha_2:\bottom}{[\alpha_1: A \limplies \bottom]}}
      \qquad
      \infer[\limpliesI]{\alpha: A \limplies B}{\infer*{\alpha: B}{[\alpha: A]}}
      \qquad
      \infer[\limpliesE]{\alpha: B}{\alpha: A \limplies B & \alpha: A}
      \\
      \infer[\gI]{\alpha b_1: \g A}{\infer*{\alpha b_1b_2:A}{[b_1 \descb b_2]}}
      \qquad
      \infer[\gE]{\alpha b_1 b_2: A}{\alpha b_1: \g A & b_1 \descb b_2}
      \qquad
      \infer[\ser]{\alpha: A}{\infer*{\alpha : A}{[b_1 \nextb b_2]}}
      \qquad
        \infer[\lin]{\alpha: A}{b_1 \nextb b_2 & b_1 \nextb b_3 & \varphi & \infer*{\alpha:A}{[\varphi[b_3/b_2]]}}
      \\
      \infer[\xI]{\alpha b_1: \x A}{\infer*{\alpha b_1b_2:A}{[b_1 \nextb b_2]}}
      \qquad
      \infer[\xE]{\alpha b_1b_2: A}{\alpha b_1: \x A & b_1 \nextb b_2}
      \qquad
      \infer[\refldesc]{\alpha : A}{\infer*{\alpha : A}{[b_1 \descb b_1]}}
      \qquad
      \infer[\transdesc]{\alpha : A}
       {
       b_1 \descb b_2 
       & 
       b_2 \descb b_3 
       & 
       \infer*{\alpha : A}{[b_1 \descb b_3]}
       }
     \\
       \infer[\historyI]{\alpha b_1 b_3: \history A}{\infer*{\alpha b_1 b_2:A}{[b_1 \descb b_2]\, [b_2\descb b_3]}}
      \qquad
      \infer[\historyE]{\alpha b_1 b_2: A}{\alpha b_1 b_3: \history A & b_1 \descb b_2 & b_2 \descb b_3}
     \qquad
     \infer[\last]{\alpha b: \local{A}}{\beta b:\local{A}}
      \qquad
     \infer[\eqdesc]{\alpha b_2:A}{b_1\descb b_2 & b_2 \descb b_1 & \alpha b_1:A}
\\
     \infer[\splitdesc]{\alpha : A}{b_1 \descb b_2 & \varphi & \infer*{\alpha :A}{[\varphi[b_2/b_1]]} & \infer*{\alpha:A}{[b_1 \nextb b'] \, [b' \descb b_2]}}
    \qquad
     \infer[\base]{\alpha : A}{b_1 \nextb b_2 & \infer*{\alpha :A}{[b_1 \descb b_2]}}
     \qquad
     \infer[\ind]{\alpha b: A}{\alpha b_0:A & b_0 \descb b & \infer*{\alpha b_j:A}{[b_0 \descb b_i] \, [b_i \nextb b_j] \, [\alpha b_i:A]}}
    \end{array}
  \end{displaymath}
The rules have the following side conditions:
\begin{itemize}
\item In $\xI$ ($\gI$), $b_2$ is \emph{fresh}, i.e.,~it is different from
$b_1$ and does not occur in any assumption on which $\alpha b_1 b_2:A$
depends other than the discarded assumption $b_1 \nextb b_2$ ($b_1
\descb b_2$).
\item In $\historyI$, $b_2$ is \emph{fresh}, i.e.,~it is different from
$b_1$ and $b_3$, and does not occur in any assumption on which $\alpha
b_1 b_2:A$ depends other than the discarded assumptions $b_1 \descb b_2$
and $b_2 \descb b_3$.
\item In $\last$, the formula must be of the form $\local{A}$, as defined  in \eqref{eq:local-language}.
\item In $\ser$, $b_2$ is fresh, i.e.,~it is different from $b$ and does
not occur in any assumption on which $\alpha:A$ depends other than the
discarded assumption $b_1 \nextb b_2$.
\item In $\splitdesc$, $b'$ is fresh, i.e.,~it is different from $b_1$
and $b_2$ and does not occur in any assumption on which $\alpha:A$
depends other than the discarded assumptions $b_1 \nextb b'$ and $b'
\descb b_2$.
\item In $\ind$, $b_i$ and $b_j$ are fresh, i.e.,~they are different from
each other and from $b$ and $b_0$, and do not occur in any assumption on
which $\alpha b_0 b_j:A$ depends other than the discarded assumptions of
the rule.
\end{itemize}
\caption{The rules of $\nltlh$} \label{fig:rules} 
\end{figure}

The rules of $\nltlh$ are given in Figure~\ref{fig:rules}. 
In $\nltlh$ we do not make use of a proper relational
labeling algebra (as, e.g., in~\cite{Vig:Labelled:00}) that contains
rules that derive rwffs from other rwffs or even lwffs. Since we are
mainly interested in the derivation of logical formulas, we rather
follow an approach that aims at simplifying the system: we use
rwffs only as assumptions for the derivation of
lwffs (as in Simpson's system for intuitionistic modal
logic~\cite{Sim:PhdThesis:94}). Thus, in $\nltlh$ there are no rules
whose conclusion is an rwff.

The rules $\limpliesI$ and $\limpliesE$ are just the labeled version of
the standard~\cite{Pra:NatDed:65} natural
deduction rules for implication introduction and elimination, where the
notion of \emph{discharged/open assumption} is also standard; e.g.,~$[\alpha:A]$ means that the formula is discharged in
$\limpliesI$. The rule $\bottomE$ is a labeled version of \emph{reductio
ad absurdum}, where
we do not constrain the time instant sequence
($\alpha_2$) in which we derive a contradiction to be the same
($\alpha_1$) as in the assumption.

The rules for the introduction and the elimination of $\g$ and $\x$
share the same structure since they both have a ``universal''
formulation. Consider, for instance, $\g$ and the corresponding relation
$\descb$. The idea underlying the introduction rule $\gI$ is that the
meaning of $\alpha b_1: \g A$ is given by the metalevel implication $b_1
\descb b_2 \Longrightarrow \alpha b_1 b_2:A$ for an arbitrary $b_2$
$\descb$-accessible from $b_1$ (where the arbitrariness of $b_2$ is
ensured by the side-condition on the rule). As we remarked above, the
operators $\g$ and $\x$ have a local nature, in that when we write
$\alpha b_1:\g A$ (and similarly for $\alpha b_1: \x A$) we are stating
that $\g A$ holds at time instant $b_1$, which is the last in the
sequence $\alpha b_1$. Hence, the elimination rule $\gE$ says that if
$b_2$ is $\descb$-accessible from $b_1$ (i.e.,~$b_1 \descb b_2$), then we
can conclude that $A$ holds for the sequence $\alpha b_1 b_2$. Similar
observations hold for $\x$ and the corresponding relation $\nextb$.

The rule $\ser$ models the fact that every time instant has an immediate
successor, while the rule $\lin$ specifies that such a successor must be
unique. $\ser$ tells us that if assuming $b_1 \nextb b_2$ we can derive
$\alpha:A$, then we can discharge the assumption and conclude that
indeed $\alpha:A$. $\lin$ is slightly more complex: assume that $b_1$
had two different immediate successors $b_2$ and $b_3$ (which we know cannot be) and assume
that the generic formula $\varphi$ holds; if by substituting $b_3$ for
$b_2$ in $\varphi$ we obtain $\alpha:A$, then we can discharge the
assumption and conclude that indeed $\alpha:A$.

Similarly, the rules $\refldesc$ and $\transdesc$ state  the
reflexivity and transitivity of $\descb$, while $\eqdesc$ captures
substitution of equals.\footnote{Recall that in this paper we use rwffs
only as assumptions for the derivation of lwffs, so we do not need a 
more general rule that concludes $\varphi[b_2/b_1]$ from $\varphi$, $b_1\descb b_2$ and $b_2 \descb b_1$.
}
The rule $\splitdesc$ states that if $b_1 \descb b_2$, then either $b_1 =
b_2$ or $b_1 < b_2$. The rule thus works in the style of a disjunction
elimination: if by assuming either of the two cases, we can derive a
formula $\alpha:A$, then we can discharge the assumptions and conclude
$\alpha:A$. Since we do not use $=$ and $<$ explicitly in our syntax, we
express such relations in an indirect way: the equality of $b_1$
and $b_2$ is expressed by 
replacing one with the
other in a generic formula $\varphi$, $<$ by the composition of $\nextb$
and $\descb$.

The rule $\base$ expresses the fact that $\descb$ contains $\nextb$,
while the rule
$\ind$ models the induction principle underlying the relation between
$\nextb$ and $\descb$. If (base case) $A$ holds in $\alpha b_0$ and if
(inductive step) by assuming that $A$ holds in $\alpha b_i$ for an
arbitrary $b_i$ $\descb$-accessible from $b_0$, we can derive that $A$
holds also in $\alpha b_j$, where $b_j$ is the immediate successor of
$b_i$, then we can conclude that $A$ holds in every $\alpha b$ such that
$b$ is $\descb$-accessible from $b_0$.\footnote{The rule is given only
in terms of relations between labels, since we
restrict the treatment of operators in the system to the
specific rules for their introduction and elimination.}

Finally, we have three rules that speak about the history and the label
sequences: the rules $\historyI$ and $\historyE$, which
we already described in the introduction, and  $\last$. This
rule expresses what we also anticipated in the introduction: the
standard operators (and connectives) of $\ltl$ speak not about sequences
of time instants but about single time instants, and thus if a
formula $A$ whose outermost operator is not $\history$ holds at $\beta
b$, then we can safely replace 
$\beta$ by any other sequence $\alpha$ and conclude that $A$ holds at
$\alpha b$. To formalize this, we define the set of (well-formed)
\emph{$\ltllocal$-formulas} (denoted by $\local{A}$) by means of the grammar
	\begin{align}
	\local{A} & \ \, ::= \ \, p \mid \bottom \mid (\local{A}) \limplies (\local{A}) \mid \g (\hist{A}) \mid \x (\hist{A}) 
 \label{eq:local-language}	\\
	\hist{A} & \ \, ::= \ \, \local{A} \mid (\hist{A}) \limplies (\hist{A}) \mid \history (\hist{A})\notag
\end{align}
where $p$ is a propositional symbol. Hence, in a formula $\local{A}$,
the history operator $\history$ can only appear in the scope of a
temporal operator $\g$ (and thus of $\f$ as in the translation
\eqref{eq:translation}) or $\x$. The rule $\last$ applies to these
formulas only; in fact, the ``$l$'' in $\local{A}$ stands for ``last'',
but it also conveniently evokes both ``local'' and ``$\ltl$''. For
formulas $\history A$ whose outermost operator is the history operator
$\history$, such a rule does not make sense (and in fact is not sound)
as it would mean changing the interval over which $A$ holds.

Such considerations are formalized in the following lemma, where we
prove, for $\ltllocal$-formulas, a result that is the analogous of the
one given in Lemma \ref{lem:last} with respect to the translation of
$\ltl$-formulas.\footnote{In fact, 
Lemma~\ref{lem:last} is a direct consequence of Lemma~\ref{lem:last-Al}
and of Lemma~\ref{lem:tr-ltl-is-ltllocal} below.} 
At the same time, we also prove that if $A$ is a formula belonging to the syntactic category $\hist{A}$ of the grammar \eqref{eq:local-language} (we will call such formulas \emph{$\ltlhistory$-formulas}), then the truth value of $A$ depends on at most the last two elements of an observation sequence.

\begin{lemma}
	\label{lem:last-Al}
	Let $\m$ be an $\ltl$-model, $[n_1, \ldots, n_k]$ an observation sequence, $\local{A}$ an $\ltllocal$-formula and $\hist{A}$ an $\ltlhistory$-formula. Then: $(i)$
            $\m, [n_1, \ldots, n_k] \modelsltlh \local{A}$ iff $\m, [m_1, \ldots, m_r, n_k] \modelsltlh \local{A}$ for every sequence $m_1, \ldots, m_r \,$ and $(ii)$
            $\m, [n_1, \ldots,$ $n_{k-1},$ $n_k]$ $\modelsltlh \hist{A}$ iff $\m, [m_1, \ldots, m_r, n_{k-1}, n_k] \modelsltlh \hist{A}$ for every sequence $m_1, \ldots, m_r$.
\end{lemma}

Given the rules in Fig.~\ref{fig:rules}, the notions of
\emph{derivation}, \emph{assumption} (\emph{open} or \emph{discharged},
as we remarked) and \emph{conclusion} are the standard ones for natural
deduction systems~\cite{Pra:NatDed:65}. We
write $\Phi \vdashltlh \alpha:A$ to say that there exists a derivation
of $\alpha:A$ in the system $\nltlh$ whose open assumptions are all
contained in the set of formulas $\Phi$. A derivation of $\alpha:A$ in
$\nltlh$ where all the assumptions are discharged is a \emph{proof} of
$\alpha:A$ in $\nltlh$ and we then say that $\alpha:A$ is a
\emph{theorem} of $\nltlh$ and write $\vdashltlh \alpha:A$.

To denote that $\Pi$ is a derivation of $\alpha:A$ whose set of
assumptions may contain the formulas $\varphi_1, \ldots, \varphi_n$, we write
\begin{displaymath}\small
  \vcenter{
     	\deduce{\alpha:A}{\deduce{\Pi}{\varphi_1 \ldots \varphi_n}}
  }
\end{displaymath}
%

If we are interested in $\ltl$-reasoning, then we can restrict our
attention to a subset of the $\nltlh$-derivations, namely, to the
derivations where the conclusion and all the open assumptions correspond
to the translations of $\ltl$-formulas.

\begin{definition}\em
  \label{def:vdashltl}
  Let $\Pi$ be a derivation in $\nltlh$ and $\Phi$ the set containing the conclusion and the open assumptions of $\Pi$. We say that $\Pi$ is an \emph{$\ltl$-derivation} iff there exists a label $b$ such that
for every $\varphi$ in $\Phi$ there exists an $\ltl$-formula $A$ such that $\varphi = b:\tr{A}$.
We write $\Gammaltl \vdashltl A$ to denote that in $\nltlh$ there exists
an $\ltl$-derivation of $b:\tr{A}$ from open assumptions in a set
$\Phi$, where $\Gammaltl = \{ B \,\mid\, b:\tr{B}\in \Phi
\}$.
\end{definition}

In Definition~\ref{def:vdashltl}, we require all the open assumptions and
the conclusion of an $\ltl$-derivation to be lwffs labeled by the same
single label $b$. Note that, as a consequence of
Corollary~\ref{cor:last}, we would obtain the same notion of
$\ltl$-derivation by requiring instead that such formulas were labeled
by the same sequence $\alpha$.

In Section~\ref{sub:sandc}, we will show that $\nltlh$ is sound
with respect to the semantics of $\ltlh$ and, by considering the notion of $\ltl$-derivability $\vdashltl$, that it is sound and weakly complete with
respect to $\ltl$. 
An investigation of
completeness with respect to $\ltlh$ is left for future work, together
with the formalization of an axiomatization of $\ltlh$.

Related to this, it is important to understand what exactly is the
relationship of the class of $\ltllocal$-formulas and the class of
$\ltl$-formulas, in particular with respect to the translation
$\tr{(\cdot)}$. It is not difficult to see that the co-domain of the
translation is included in $\ltllocal$ by construction of
$\tr{(\cdot)}$, i.e.,~by induction on the formula complexity it follows
that:

\begin{lemma}
  \label{lem:tr-ltl-is-ltllocal}
	If $A$ is an $\ltl$-formula, then $\tr{A}$ is an $\ltllocal$-formula.
\end{lemma}

The other direction is trickier, as it basically amounts to defining an
inverse translation. To solve this problem, 
we have been considering normal forms of $\ltllocal$-formulas 
and we conjecture that the following fact indeed holds.
\begin{conjecture}
 If $A$ is an $\ltllocal$-formula, then there exists an $\ltl$-formula $B$ such that $B^*$ is semantically equivalent to $A$.
\end{conjecture}

\subsection{Soundness and completeness}\label{sub:sandc}


\begin{theorem}\label{th:ltlh-soundness}
For every set $\Phi$ of labeled and relational formulas and every
labeled formula $\alpha:A$, if $\Phi \vdashltlh \alpha:A$, then $\Phi
\modelsltlh \alpha:A$.
\end{theorem}
\begin{proof}
The proof proceeds by induction on the structure of the derivation of
$\alpha:A$. The base case is when $\alpha:A \in \Phi$ and is trivial.
There is one step case for every rule and we show here only the two
representative cases
\begin{displaymath}\footnotesize
\infer[\historyI]{\beta b_1 b_3:\history B}
 {
 \deduce{\beta b_1 b_2:B}{\deduce{\Pi}{[b_1 \descb b_2] & [b_2 \descb b_3]}}
 }
\qquad \quad {\normalsize \text{and}} \qquad \quad
\infer[\last]{\beta b: A}
 {
  \deduce{\beta' b:A}{
      \Pi
    }
 }
\end{displaymath}
\begin{SHORT}
Some more cases are in
\cite{}.
\end{SHORT}
\begin{LONG}
Some more cases are in
Appendix~\ref{ap:ltlh-soundness}.
\end{LONG}
First, consider the case in which the last rule application is a $\historyI$, where $\alpha = \beta b_1 b_3$, $A=\history B$, and 
$\Pi$ is a proof of $\beta b_1 b_2: B$ from hypotheses in $\Phi'$, with $b_2$ fresh and with $\Phi' = \Phi \cup \{b_1 \descb b_2\} \cup \{b_2 \descb b_3\}$. By the induction hypothesis, for every interpretation $\interp$, if $\m, \interp \modelsltlh \Phi'$, then $\m, \interp \modelsltlh  \beta b_1 b_2:B$. We let $\interp$ be any interpretation such that $\m, \interp \modelsltlh \Phi$, and show that $\m, \interp \modelsltlh \beta b_1 b_3:\history B$. Let $\interp(b_1) = n$, $\interp(b_3) = m$ and $\interp^+(\beta) = [n_1, \ldots, n_k]$. Since $b_2$ is fresh, we can extend $\interp$ to an interpretation (still called $\interp$ for simplicity) such that $\interp(b_2)=n+i$ for an arbitrary $0 \le i \le m$. The induction hypothesis yields $\m, \interp \modelsltlh \beta b_1 b_2:B$, i.e.,~$\m , [n_1, \ldots, n_k, n, n+i]\modelsltlh B$, and thus, since $i$ is an arbitrary point between $0$ and $m$, we obtain $\m, [n_1, \ldots, n_k, n, n+m] \modelsltlh \history B$. It follows $\m, \interp \modelsltlh \beta b_1 b_3: \history B$.

Now consider the case in which the last rule applied is $\last$ and $\alpha = \beta b$, where $\Pi$ is a proof of $\beta' b: A$ from hypotheses in $\Phi$. By applying the induction hypothesis on $\Pi$, we have $\Phi \modelsltlh \beta' b: A$.
We proceed by considering a generic $\ltl$-model $\m$ and a generic interpretation $\interp$ on it such that $\m, \interp \modelsltlh \Phi$ and showing that this entails $\m, \interp \modelsltlh \beta b:A$.
By the induction hypothesis, $\m, \interp \modelsltlh \beta' b: A$, i.e.,~$\m, \interp^+(\beta' b) \modelsltlh A$ by~Definition \ref{def:ltlh-truth}. Since $A$ is an $\ltllocal$-formula by the side condition of the rule and the two observation sequences $\interp^+(\beta' b)$ and $\interp^+(\beta b)$ share the same last element $\interp(b)$, we can apply Lemma~\ref{lem:last-Al} and obtain $\m, \interp^+(\beta b) \modelsltlh A$, i.e.,~$\m, \interp \modelsltlh \beta b:A$ by~Definition \ref{def:ltlh-truth}.
\end{proof}        

By exploiting the translation of Section \ref{sub:translation} and the notion of $\ltl$-derivation of Definition \ref{def:vdashltl}, we also prove a result of soundness with respect to $\ltl$.
\begin{theorem}\label{th:ltl-soundness}
  For every set $\Gammaltl$ of $\ltl$-formulas and every $\ltl$-formula $A$, if $\Gammaltl \vdashltl A$, then $\Gammaltl \modelsltl A$.
\end{theorem}
\begin{proof}
  By definition of $\vdashltl$, for a given label $b$, there exists a derivation in $\nltlh$ of $b:\tr{A}$ from open assumptions in $\Phi = \{ b:\tr{B} \,\mid \, B \in \Gammaltl \}$. By Theorem \ref{th:ltlh-soundness}, $\Phi \vdashltlh b:\tr{A}$ implies $\Phi \modelsltlh b:\tr{A}$. Since $b$ is generic, we have that for every $\ltl$-model $\m$ and every interpretation $\interp$, $\m,\interp \modelsltlh \Phi$ implies $\m,\interp \modelsltlh b:\tr{A}$ iff for every natural number $n$, $\m, [n] \modelsltlh \tr{\Gammaltl}$ implies $\m,[n] \modelsltlh \tr{A}$, where $\tr{\Gammaltl}=\{ \tr{B} \,\mid\, B\in\Gammaltl \}$. By Lemma \ref{lem:last}, we infer that 
for every observation sequence $\sigma$, $\m, \sigma \modelsltlh \tr{\Gammaltl}$ implies $\m,\sigma \modelsltlh \tr{A}$. By Definition \ref{def:ltlh-truth}, $\tr{\Gammaltl} \modelsltlh \tr{A}$ and thus, by Theorem \ref{th:translation}, we conclude $\Gammaltl \modelsltl A$.
%
\end{proof}


As we anticipated, an analysis of the completeness of $\nltlh$ with
respect to $\ltlh$ is left for future work. Here we discuss completeness
with respect to $\ltl$. The proposed natural deduction system consists
of only finitary rules; consequently, it cannot be strongly complete for
$\ltl$.\footnote{This is not a problem of our formulation: all the
finitary deduction systems for temporal logics equipped with at least
the operators $\x$ and $\g$ have such a defect; see, e.g.,
\cite[Ch.~6]{Kro:TempLogProgr:87:short}.} 
Nevertheless, by using the axiomatization $\axltl$ and the translation $\tr{(\cdot)}$, we can give a proof of weak completeness for it; namely:

\begin{theorem}\label{th:ltl-completeness}
For every $\ltl$-formula $A$,
if $\modelsltl A$, then $\vdashltl A$.
%
%
\end{theorem}
\begin{proof}
We can prove the theorem by showing that $\nltlh$
is complete with respect to the axiomatization $\axltl$ given in
\secref{sec:ltl}, which is sound and complete for the logic
$\ltl$. That is, we need to prove that:
(i) the translation, via $\tr{(\cdot)}$, of every axiom of $\axltl$
is provable in $\nltlh$ by means of an $\ltl$-derivation, and (ii) the notion of $\vdashltl$ is closed under the
(labeled equivalent of the) rules of inference of $\axltl$. 
Showing (ii) is straightforward and we omit it here. As an example for (i), 
we give here a derivation of the translation of $(A6)$.
%
\begin{SHORT}
The other cases are presented in \cite{}.
\end{SHORT}
\begin{LONG}
The other cases are presented in Appendix \ref{ap:ltl-completeness}.
\end{LONG}

  \begin{displaymath}
\vcenter{
\scriptsize
\infer[\limpliesI^1]{ b: \g( A \limplies \x  A) \limplies ( A \limplies \g  A)}
 {
 \infer[\limpliesI^2]{ b:  A \limplies \g  A}
  {
  \infer[\limplies^3]{ b: \g  A}
   {
   \infer[\last]{ bc:  A}
   {
   \infer[\ind^4]{ c:  A}
    {
    [ b:  A]^2
    &
    [b\descb c]^3
    &
    \infer[\last]{ b_j: A}
     {
    \infer[\xE]{ b b_i b_j: A}
     {
     \infer[\limpliesE]{ b b_i: \x  A}
      {
      \infer[\gE]{ b b_i:  A \limplies \x  A}
       {
       [ b:\g ( A \limplies \x  A)]^1
       &
       [b \descb b_i]^4
       }
      &
      \infer[\last]{ b b_i: A}{[ b_i: A]^4}
      }
     &
     [b_i \nextb b_j]^4
     }
    }
   }
  }
}
 }
}
}
\end{displaymath}
\end{proof}

\section{Conclusions}
\label{sec:discussion}

The introduction of the operator $\history$ has allowed us to formalize
the ``history'' of until and thus, via a proper translation, to give a
labeled natural deduction system for a linear time logic endowed with
$\history$ 
that is also sound and complete with respect to $\ltl$ with until. As we
remarked above, we see this work as spawning several different
directions for future research. First, the ``recipe'' for dealing with
until that we gave here is abstract and general, and thus provides the
basis for formalizing deduction systems for temporal logics endowed with
$\until$, both linear and branching time. We are
currently considering $\ctlstar$ and its sublogics as
in~\cite{MasiniViganoVolpe09a,Rey:TableauBundled:07} and are also
working at a formal characterization of the class of logics that can be
captured with our approach.

Second, the well-behaved nature of our approach,
where
each connective and operator has one introduction and one elimination
rule, paves the way to a proof-theoretical analysis of the resulting
natural deduction systems, e.g.,~to show proof normalization and other
useful meta-theoretical properties, which we are tackling in current work.
Moreover, we are also considering different optimizations of the rules.
In particular, along the lines of the discussion about the rule $\last$
(and Corollary~\ref{cor:last} and Definition~\ref{def:vdashltl}), we are
investigating to what extent we can use sequences as labels only when
they are really needed, which would also simplify the proofs of
normalization and other meta-properties\footnote{As an interesting
side-track, we believe that the restrictions we imposed on formulas for
the rule $\last$, i.e.,~considering $\local{A}$ and $\hist{A}$, is
closely related, at least in spirit, to the focus on \emph{persistent}
formulas when combining intuitionistic and classical logic so as to
avoid the collapse of the two logics into one,
see~\cite{FarinasHerzig96short} but
also~\cite{CaleiroRamos07,GabbayFrocos96}. We are, after all,
considering here formulas stemming from two classes (if not two logics
altogether), and it makes thus sense that they require different
labeling (single instants and sequences).}.

This is closely related to the formalization of the relationship between
the class of $\ltllocal$-formulas and that of $\ltl$-formulas, which in
turn will allow us to reason about the completeness of $\nltlh$ with
respect to the semantics of $\ltlh$ and also to provide an
axiomatization of $\ltlh$ (thus treating it as a full-fledged logic as
opposed to as a ``service'' logic for $\ltl$ as we did here).

Finally, it is worth observing that several works have considered
\emph{interval temporal logics},
e.g.,~\cite{BowmanThompson03,CerritoCialdeaMayer00,Goranko06,HalpernShoham91,SGL:PLTL1}.
While these works consider intervals explicitly, we have used them
somehow implicitly here, as a means to formalize the dual nature of
until via the history $\history$, and this is another reason
why it is interesting to reduce the use of label sequences as much as
possible.
A more detailed comparison of our approach with these works is left for
future work.



\medskip
\noindent
{\bf Acknowledgments}
This work was partially supported by the PRIN projects ``CONCERTO'' and ``SOFT''.


\begin{LONG}
\appendix
\newpage

\section{Proofs}

\subsection{Properties of the translation $\tr{(\cdot)}$}
    \label{ap:translation}

\paragraph{Proof of Lemma \ref{lem:last}}
	By induction on the complexity of $A$. The base case is when $A= p$ or $A = \bottom$ and is trivial. There is one inductive step case for each connective and temporal operator.
	\begin{description}
		\item[$A = B \limplies C$.]
			Then the translation of $A$ is $\tr{A} = \tr{B} \limplies \tr{C}$. By Definition \ref{def:ltlh-truth}, we obtain $\m, [n_1, \ldots, n_k] \modelsltlh \tr{B} \limplies \tr{C}$ iff $\m, [n_1, \ldots, n_k] \modelsltlh \tr{B} \,$ implies  $\, \m, [n_1, \ldots, n_k] \modelsltlh \tr{C}.$ 
			By the induction hypothesis, we see that this holds iff $\m, [m_1, \ldots, m_r, n_k] \modelsltlh \tr{B}$ implies $\m, [m_1, \ldots, m_r, n_k] \modelsltlh \tr{C}$ for every sequence $m_1, \ldots, m_r$ and thus, by Definition \ref{def:ltlh-truth}, iff for every sequence $m_1, \ldots, m_r$, $\; \m, [m_1, \ldots, m_r, n_k] \modelsltlh \tr{B} \limplies \tr{C}$.
		\item[$A = \g B$.]
			Then $\tr{A} = \g \tr{B}$. By Definition \ref{def:ltlh-truth}, $\m, [n_1, \ldots, n_k] \modelsltlh \g \tr{B}$ iff $\forall m \ge n_k. \, \m, [n_1, \ldots, n_k, m] \modelsltlh \tr{B}$ iff (by the induction hypothesis) $\forall m \ge n_k.\, \m,$ $[m_1, \ldots, m_r, n_k, m] \modelsltlh \tr{B}$ for every sequence $m_1, \ldots, m_r$ iff (by Definition \ref{def:ltlh-truth}) $\m, [m_1, \ldots, m_r, n_k] \modelsltlh \g \tr{B}$, for every sequence $m_1, \ldots, m_r$.
		\item[$A = \x B$.]
			This case is very similar to the previous one and we omit it.
		\item[$A = B \until C$.]
			Then $\tr{A} = \tr{C} \vee (\f(\x \tr{C} \wedge \history \tr{B}))$. By Definition \ref{def:ltlh-truth}, we have $\m,$ $[n_1,$ $\ldots,$ $n_k] \modelsltlh \tr{A}$ iff $(\m, [n_1, \ldots, n_k] \modelsltlh \tr{C}$ or $\m, [n_1, \ldots, n_k]$ $\modelsltlh \f(\x\tr{C} \wedge \history \tr{B}))$ iff $(\m, [n_1, \ldots, n_k]$ $\modelsltlh \tr{C}$ or $\exists m \ge n_k. \, (\m, [n_1, \ldots, n_k, m]$ $\modelsltlh \x\tr{C} \wedge \history \tr{B}))$ iff $(\m, [n_1, \ldots, n_k]$ $\modelsltlh \tr{C}$ or $\exists m \ge n_k. \, (\m, [n_1, \ldots, n_k, m]$ $\modelsltlh \x\tr{C}$ and $\m, [n_1, \ldots, n_k, m] \modelsltlh \history \tr{B}))$ iff $(\m, [n_1, \ldots, n_k] \modelsltlh \tr{C}$ or $\exists m \ge n_k. \, (\m, [n_1, \ldots, n_k, m, m+1] \modelsltlh \tr{C}$ and $\forall l.\, n_k \le l \le m$ implies $\m, [n_1, \ldots,$ $ n_k, l] \modelsltlh \tr{B}))$ iff (by the induction hypothesis) for every sequence $m_1, \ldots, m_r$, we have $(\m, [m_1, \ldots, m_r, n_k] \modelsltlh \tr{C}$ or $\exists m \ge n_k. \, (\m, [m_1, \ldots, m_r, n_k, m, m+1] \modelsltlh \tr{C}$ and $\forall l.\, n_k \le l \le m$ implies $\m, [m_1, \ldots, m_r, n_k, m, l] \modelsltlh \tr{B}))$ iff (by Definition \ref{def:ltlh-truth}) $\m, [m_1, \ldots, m_r, n_k]$ $\modelsltlh \tr{C} \vee (\f(\x \tr{C} \wedge \history \tr{B}))$ for every sequence $m_1, \ldots, m_r$.			
	\end{description}
\qed
%
\paragraph{Proof of Corollary \ref{cor:last}}
	Immediate, by Lemma \ref{lem:last}.
\qed

\paragraph{Proof of Lemma \ref{lem:translation}}
	By induction on the complexity of $A$. The base case is when $A= p$ or $A = \bottom$ and is trivial. As inductive step, we have a case for each connective and temporal operator.
	\begin{description}
		\item[$A = B \limplies C$.]
			Then $\tr{A} = \tr{B} \limplies \tr{C}$. We have $\m, n \modelsltl B \limplies C$ iff (by Definition \ref{def:ltl-truth}) $\m, n \modelsltl B$ implies $\m, n \modelsltl C \;$ iff (by the induction hypothesis) $\m, [n] \modelsltlh \tr{B}$ implies $\m, [n] \modelsltlh \tr{C}$ iff (by Definition \ref{def:ltlh-truth}) $\m,[n] \modelsltlh \tr{B} \limplies \tr{C}$.
		\item[$A = \g B$.]
			Then $\tr{A} = \g \tr{B}$. We have $\m, n \modelsltl \g B$ iff (by Definition \ref{def:ltl-truth}) $\forall m \ge n. \, \m, m \modelsltl B$ iff (by the induction hypothesis) $\forall m \ge n.\, \m,[m] \modelsltlh \tr{B}$ iff (by Lemma \ref{lem:last}) $\forall m \ge n. \, \m,[n,m] \modelsltlh \tr{B}$ iff (by Definition \ref{def:ltlh-truth}) $\m,[n] \modelsltlh \g \tr{B}$.
		\item[$A = \x B$.]
			This case is very similar to the previous one and we omit it.			
		\item[$A = B \until C$.]
			Then $\tr{A} = \tr{C} \vee (\f(\x \tr{C} \wedge \history \tr{B}))$. We have $\m, n \modelsltl A$ iff (by Definition \ref{def:ltl-truth}) $\exists m \ge n. \, \m, m \modelsltl C$ and $\forall n'.\, n \le n' < m$ implies $\m,n' \modelsltl B$ iff $\m,n \modelsltl C$ or $(\exists m > n. \, \m, m \modelsltl C$ and $\forall n'.\,n \le n' < m$ implies $\m,n' \modelsltl B)$ iff (by the induction hypothesis) $\m,[n] \modelsltlh \tr{C}$ or $(\exists m > n. \, \m, [m] \modelsltlh \tr{C}$ and $\forall n'.\,n \le n' < m$ implies $\m,[n'] \modelsltlh \tr{B})$ iff (by Lemma \ref{lem:last}) $\m,[n] \modelsltlh \tr{C}$ or $(\exists m > n. \, \m, [n, m] \modelsltlh \tr{C}$ and $\forall n'.\,n \le n' < m$ implies $\m,[n,n'] \modelsltlh \tr{B})$ iff $\m,[n] \modelsltlh \tr{C}$ or $(\exists l \ge n. \, \m, [n, l, l+1] \modelsltlh \tr{C}$ and $\forall n'.\,n \le n' \le l$ implies $\m,[n,n'] \modelsltlh \tr{B})$ iff (by Definition \ref{def:ltlh-truth}) $\m,[n]\modelsltlh \tr{C}$ or $(\exists l\ge n.\,$ $\m,[n,l] \modelsltlh \x \tr{C} \wedge \history \tr{B})$ iff (by Definition~\ref{def:ltlh-truth}) $\m,[n]\modelsltlh \tr{C} \vee \, \f (\x \tr{C} \wedge \history \tr{B}).$
	\end{description}	
\qed

\subsection{The system $\nltlh$}
  \label{ap:nltlh}

  \paragraph{Proof of Lemma \ref{lem:last-Al}}
The proofs of the statements $(i)$ and $(ii)$ proceed in parallel and are by induction on the formula complexity. The base case is when $\local{A}= p$ or $\local{A} = \bottom$ and is trivial. There is one inductive step case for each other formation case coming from the recursive definition of the grammar \eqref{eq:local-language}. Along the proof, $\local{A}, \local{B}, \local{C},\ldots$ denote $\ltllocal$-formulas while $\hist{A}, \hist{B}, \hist{C},\ldots$ denote $\ltlhistory$-formulas.
	\begin{description}
		\item[$\local{A} = \local{B} \limplies \local{C}$.]
			By Definition \ref{def:ltlh-truth}, we have $\m, [n_1, \ldots, n_k] \modelsltlh \local{B} \limplies \local{C}\;$ iff $\;\m,$ $[n_1,$ $\ldots, n_k] \modelsltlh \local{B} \,$ implies  $\, \m, [n_1, \ldots, n_k] \modelsltlh \local{C}.$ 
			By the induction hypothesis, we see that this holds iff $\;\m, [m_1, \ldots, m_r,$ $n_k] \modelsltlh \local{B} \, \mbox{ implies } \, \m, [m_1, \ldots,$ $m_r,$ $n_k] \modelsltlh \local{C}\;$ for every sequence $m_1, \ldots, m_r$ and thus, by Definition \ref{def:ltlh-truth}, iff for every sequence $m_1, \ldots, m_r$, $\; \m,$ $[m_1, \ldots, m_r,$ $n_k]$ $\modelsltlh \local{B} \limplies \local{C}\;$.
		\item[$\local{A} = \g \hist{B}$.]
			$\m, [n_1, \ldots, n_k] \modelsltlh \g \hist{B}\;$ iff (by Definition \ref{def:ltlh-truth}) $\;\forall m \ge n_k. \, \m,$ $[n_1, \ldots, n_k, m] \modelsltlh \hist{B} \;$ iff (by the induction hypothesis) $\; \forall m \ge n_k.\, \m,$ $[m_1, \ldots, m_r, n_k, m] \modelsltlh \hist{B} \,$ for every sequence $m_1, \ldots, m_r \;$ iff (by Definition \ref{def:ltlh-truth}) $\; \m, [m_1, \ldots, m_r, n_k] \modelsltlh \g \hist{B} \,$ for every sequence $m_1, \ldots, m_r$.
		\item[$\local{A} = \x \hist{B}$.]
			This case is very similar to the previous one and we omit it.
		\item[$\hist{A} = \local{B}$.]
			$\m, [n_1, \ldots, n_k] \modelsltlh \local{B} \;$ iff (by the induction hypothesis) $\; \m,$ $[i_1, \ldots,$ $i_s,n_k] \modelsltlh \local{B}\;$ for every sequence $i_1 \ldots, i_s$ and thus also $\; \m,$ $[m_1, \ldots,$ $m_r, n_{k-1}, n_k] \modelsltlh \local{B}\;$ for every sequence $m_1, \ldots, m_r$.
                \item[$\hist{A} = \hist{B} \limplies \hist{C}$.]
			 $\m, [n_1, \ldots, n_k] \modelsltlh \hist{B} \limplies \hist{C}\;$ iff (by Definition~\ref{def:ltlh-truth}) $\;\m, [n_1,$ $\ldots, n_k] \modelsltlh \hist{B} \,$ implies  $\, \m, [n_1, \ldots, n_k] \modelsltlh \hist{C}.$ 
			By the induction hypothesis, this holds iff $\;\m, [m_1, \ldots, m_r, n_{k-1}, n_k] \modelsltlh \hist{B} \, \mbox{ implies } \, \m,$ $[m_1,$ $\ldots,$ $m_r, n_{k-1}, n_k] \modelsltlh \hist{C}\;$ for every sequence $m_1, \ldots, m_r$ and thus, by Definition \ref{def:ltlh-truth}, iff for every sequence $m_1, \ldots, m_r$, $\; \m, [m_1, \ldots, m_r, n_{k-1}, n_k]$ $\modelsltlh \hist{B} \limplies \hist{C}\;$.
                \item[$\hist{A} = \history \hist{B}$.]
			$\m, [n_1, \ldots, n_k] \modelsltlh \history \hist{B}\;$ iff (by Definition \ref{def:ltlh-truth}) $\;\forall n.\, n_{k-1}\le n \le n_k \,\mbox{ implies }  \, \m, [n_1, \ldots, n_{k-1}, n] \modelsltlh \hist{B} \;$ iff (by the induction hypothesis) $\;\forall n.\, n_{k-1}\le n \le n_k \,\mbox{ implies }  \, \m, [m_1, \ldots, m_r, n_{k-1}, n] \modelsltlh \hist{B} \;$ for every sequence $m_1, \ldots, m_r \;$ iff (by Definition \ref{def:ltlh-truth}) $\; \m, [m_1, \ldots, m_r, n_{k-1}, n_k] \modelsltlh \history \hist{B} \,$ for every sequence $m_1, \ldots, m_r$.
	\end{description}
\qed

\subsection{Soundness}
  \label{ap:ltlh-soundness}

  \paragraph{Proof of Theorem \ref{th:ltlh-soundness}}

  We present here some more cases related to the proof of Theorem \ref{th:ltlh-soundness}, which states the soundness of the system $\nltlh$ with respect to the semantics of $\ltlh$.

  Consider the case in which the last rule application is a $\gI$, where $\alpha = \beta b_1$ and $A=\g B$:
  \begin{displaymath}
  \infer[\gI]{\beta b_1:\g B}
  {
  \deduce{\beta b_1 b_2:B}{\deduce{\Pi}{[b_1 \descb b_2]}}
  }
  \end{displaymath}
  where $\Pi$ is a proof of $\beta b_1:\g B$ from hypotheses in $\Phi'$, with $b_2$ fresh and with $\Phi' = \Phi \cup \{b_1 \descb b_2\}$. By the induction hypothesis, for all interpretations $\interp$, if $\m, \interp \modelsltlh \Phi'$, then $\m, \interp \modelsltlh  \beta b_1 b_2:B$. We let $\interp$ be any interpretation such that $\m, \interp \modelsltlh \Phi$, and show that $\m, \interp \modelsltlh \beta b_1:\g B$. Let $\interp(b_1) = n$ and $\interp^+(\beta) = [n_1, \ldots, n_k]$. Since $b_2$ is fresh, we can extend $\interp$ to an interpretation (still called $\interp$ for simplicity) such that $\interp(b_2)=n+m$ for an arbitrary $m>0$. The induction hypothesis yields $\m, \interp \modelsltlh \beta b_1 b_2:B$, i.e.,~$\m , [n_1, \ldots, n_k, n, n+m]\modelsltlh B$, and thus, since $m$ is arbitrary, we obtain $\m, [n_1, \ldots, n_k, n] \modelsltlh \g B$. It follows $\m, \interp \modelsltlh \beta b_1: \g B$.
  
  Now consider the case in which the last rule applied is $\gE$ and $\alpha = \beta b_1 b_2$:
  \begin{displaymath}
  \infer[\gE]{\beta b_1 b_2: A}
  {
  \deduce{\beta b_1: \g A}{\Pi}
  &
  b_1 \descb b_2
  }
  \end{displaymath}
  where $\Pi$ is a proof of $\beta b_1:\g A$ from hypotheses in $\Phi_1$, with $\Phi = \Phi_1 \cup \{ b_1 \descb b_2 \}$ for some set $\Phi_1$ of formulas. By applying the induction hypothesis on $\Pi$, we have:
  \begin{displaymath}
          \Phi_1 \modelsltlh \beta b_1:\g A \; .
  \end{displaymath}
  We proceed by considering a generic $\ltl$-model $\m$ and a generic interpretation $\interp$ on it such that $\m, \interp \modelsltlh \Phi$ and showing that this entails 
  \begin{displaymath}
          \m, \interp \modelsltlh \beta b_1 b_2:A \; .
  \end{displaymath}
  Since $\Phi_1 \subset \Phi$, we deduce $\m, \interp \modelsltlh \Phi_1$ and, from the induction hypothesis, $\m, \interp \modelsltlh \beta b_1: \g A$. Furthermore $\m, \interp \modelsltlh \Phi$ entails $\m, \interp \modelsltlh b_1 \descb b_2$. Then, by Definition \ref{def:ltlh-truth}, we obtain $\m, \interp \modelsltlh \beta b_1 b_2:A$.

  Now consider the case in which the last rule applied is $\historyE$ and $\alpha = \beta b_1 b_2$:
  \begin{displaymath}
  \infer[\historyE]{\beta b_1 b_2: A}
  {
  \deduce{\beta b_1 b_3: \history A}{\Pi}
  &
  b_1 \descb b_2
  &
  b_2 \descb b_3
  }
  \end{displaymath}
  where $\Pi$ is a proof of $\beta b_1 b_3:\history A$ from hypotheses in $\Phi_1$, with $\Phi = \Phi_1 \cup \{ b_1 \descb b_2 \} \cup \{ b_2 \descb b_3 \}$ for some set $\Phi_1$ of formulas. By applying the induction hypothesis on $\Pi$, we have:
  \begin{displaymath}
          \Phi_1 \modelsltlh \beta b_1 b_3:\history A \; .
  \end{displaymath}
  We proceed by considering a generic $\ltl$-model $\m$ and a generic interpretation $\interp$ on it such that $\m, \interp \modelsltlh \Phi$ and showing that this entails 
  \begin{displaymath}
          \m, \interp \modelsltlh \beta b_1 b_2:A \; .
  \end{displaymath}
  Since $\Phi_1 \subset \Phi$, we deduce $\m, \interp \modelsltlh \Phi_1$ and, from the induction hypothesis, $\m, \interp \modelsltlh \beta b_1 b_3: \history A$. Furthermore $\m, \interp \modelsltlh \Phi$ entails $\m, \interp \modelsltlh b_1 \descb b_2$ and $\m, \interp \modelsltlh b_2 \descb b_3$. Then, by Definition \ref{def:ltlh-truth}, we obtain $\m, \interp \modelsltlh \beta b_1 b_2:A$.

  Finally, consider the case in which the last rule applied is $\ind$ and $\alpha = \beta b$:
  \begin{displaymath}
  \infer[\ind]{\beta b:A}
  {
  \deduce{\beta b_0: A}{\Pi'}
  &
  b_0 \descb b
  &
  \deduce{\beta b_j:A}{\deduce{\Pi}{[b_0 \descb b_i] & [b_i \nextb b_j] & [\beta b_i: A]}}
  }
  \end{displaymath}
  where $\Pi$ is a proof of $\beta b_j:A$ from hypotheses in $\Phi_2$ and $\Pi'$ is a proof of $\beta b_0:A$ from hypotheses in $\Phi_1$, with $\Phi = \Phi_1 \cup \{ b_0 \descb b \}$ and $\Phi_2 = \Phi_1 \cup \{ b_0 \descb b_i \} \cup \{ b_i \nextb b_j \} \cup \{ \beta b_i:A \}$ for some set $\Phi_1$ of formulas. The side-condition on $\ind$ ensures that $b_i$ and $b_j$ are fresh in $\Pi$. Hence, by applying the induction hypothesis on $\Pi$ and $\Pi'$, we have:
  \begin{displaymath}
  \Phi_2 \modelsltlh \beta b_j:A \qquad \text{and} \qquad \Phi_1 \modelsltlh \beta b_0:A\, .
  \end{displaymath}
  We proceed by considering a generic $\ltl$-model $\m$ and a generic interpretation $\interp$ on it such that $\m, \interp \modelsltlh \Phi$ and showing that this entails 
  \begin{displaymath}
  \m, \interp \modelsltlh \beta b:A \; .
  \end{displaymath}
  First, we note that $\Phi_1 \subset \Phi$ and therefore $\m, \interp \modelsltlh \Phi$ implies $\m, \interp \modelsltlh \Phi_1$ and, by the induction hypothesis on $\Pi'$, $\m, \interp \modelsltlh \beta b_0:A$.
  Now let $\interp(b_0) = n$ for some natural number $n$. From $\m, \interp \modelsltlh \Phi$, we deduce $\m, \interp \modelsltlh b_0 \descb b$ and thus $\interp(b)=n+k$ for some $k\in \nat$. We show by induction on $k$ that $\m, \interp \modelsltlh \beta b:A$.
  As a base case, we have $k=0$; it follows that $\interp(b) = \interp(b_0)$ and thus trivially that $\m, \interp \modelsltlh \beta b_0:A$ entails $\m, \interp \modelsltlh \beta b:A$.
  Let us consider now the induction step. Given a label $b_{k-1}$ such that $\interp(b_{k-1})=n+k-1$, we show that the induction hypothesis $\m, \interp \modelsltlh \beta b_{k-1}:A$ entails the thesis $\m, \interp \modelsltlh \beta b:A$. We can build an interpretation $\interp'$ that differs from $\interp$ only in the points assigned to $b_i$ and $b_j$, namely, $\interp' = \interp[b_i \mapsto n+k-1][b_j \mapsto n+k]$. It is easy to verify that the interpretation $\interp'$ is such that the following three conditions hold:
  \begin{enumerate}[$(i)$]
  \item $\m, \interp' \modelsltlh \beta b_i:A$;
  \item $\m, \interp' \modelsltlh b_0 \descb b_i$;
  \item $\m, \interp' \modelsltlh b_i \nextb b_j$.
  \end{enumerate}
  Furthermore, the side-condition on the rule $\ind$ ensures that
  $\interp$ and $\interp'$ agree on all the labels occurring in
  $\Phi_1$, from which we can infer $\m, \interp' \modelsltlh 
  \Phi_1$. It follows $\m, \interp' \modelsltlh
  \Phi_2$ and thus, by the induction hypothesis on $\Pi$, $\m, \interp' \modelsltlh \beta b_j: A$.
  We conclude $\m, \interp' \modelsltlh \beta b:A$ by observing that $\interp'(b_j) = \interp(b)$.
\qed

\subsection{Completeness}
  \label{ap:ltl-completeness}

  \paragraph{Proof of Theorem \ref{th:ltl-completeness}}

We present here the $\nltlh$-derivations of the remaining axioms of $\axltl$. Note that, for simplicity, we use also some rules (i.e.,~$\fI$, $\fE$, $\veeI$, $\veeE$, $\wedgeI$ and $\wedgeE$) concerning derived operators. They can be easily derived from the set of rules in Figure~\ref{fig:rules}.

\paragraph{$\mathit{(A2)}$}
\begin{displaymath}
\scriptsize
\infer[\limpliesI^1]{ b:\g ( A \limplies  B) \limplies (\g  A \limplies \g  B)}
 {
 \infer[\limpliesI^2]{ b: \g  A \limplies \g  B}
  {
  \infer[\gI^3]{ b: \g  B} 
   {
   \infer[\limpliesE]{ bc: B}
    {
    \infer[\gE]{ bc: A \limplies  B}
     {
     [ b:\g( A \limplies  B)]^1
     & 
     [b\descb c]^3
     }
    &
    \infer[\gE]{ bc:  A}
     {
     [ b: \g  A]^2
     &
     [b \descb c]^3
     }
    }
   }
  }
 }
\end{displaymath}

\paragraph{$\mathit{(A3)}$}
 $(\x \neg  A \leftrightarrow \neg \x  A)$
\begin{displaymath}
\scriptsize
\begin{array}{c}
\infer[\limpliesI^1]{ b:\x \neg A \limplies \neg \x A}
 {
 	\infer[\ser^2]{ b:\neg \x A}{
 		\infer[\bottomE^3]{ b:\neg \x A}{
 			\infer[\limpliesE]{ bc:\bottom}{
 				\infer[\xE]{ bc:\neg A}{[ b:\x\neg A]^1 & [b \nextb c]^2}	
 					& 
 				\infer[\xE]{ bc:A}{[ b:\x A]^3 & [b \nextb c]^2}
 				}
 			}
 		}
 }
\\ \ \\
\infer[\limpliesI^1]{ b:\neg \x  A \limplies \x \neg  A}
 {
 \infer[\xI^2]{ b: \x \neg  A}
  {
  \infer[\limplies^3]{ bc:\neg  A}
   {
   \infer[\limpliesE]{ b:\bottom}
    {
    [ b:\neg \x  A]^1
    &
    \infer[\xI^4]{ b:\x  A}
     {
     \infer[\lin]{ bd:  A}
      {
      [b \nextb c]^2
      & 
      [b \nextb d]^4
      &
      [ bc: A]^3
      }
     }
    }
   }
  }
 }
\end{array}
\end{displaymath}

\paragraph{$\mathit{(A4)}$} This proof is similar to the one for
$\mathit{(A2)}$ and we thus omit it.

\paragraph{$\mathit{(A5)}$}
\begin{displaymath}
\scriptsize
\vcenter{
\infer[\limpliesI^1]{ b:\g A \limplies (A \wedge \x \g A)}
 {
 \infer[\wedge I]{ b:A \wedge \x \g A}
  {
  \infer[\refldesc^2]{ b:A}
   {
   \infer[\last]{ b:A}
   {
   \infer[\gE]{ b b:A}
    {
    [ b: \g A]^1
    &
    [b\descb b]^2
    }
   }
   }
  &
  \infer[\xI^3]{ b:\x \g A}
   {
   \infer[\gI^4]{ bc:\g A}
    {
    \infer[\last]{ bcd: A}
    {
    \infer[\basedesc^5]{ bd:A}
     {
     [b \nextb c]^3
     &
     \infer[\transdesc^6]{ bd:A}
      {
      [b \descb c]^5
      &
      [c \descb d]^4
      &
      \infer[\gE]{ bd:A}
       {
       [ b: \g A]^1
       & 
       [b \descb d]^6
       }
      }
     }
    }
   }
  }
 }}
}
\end{displaymath}

\paragraph{$\mathit{(A7)}$}
Note that, for brevity, we give here a derivation of a, clearly equivalent, simplified version of the translation of $\mathit{(A7)}$. Namely, we consider $\f(\x B \wedge \history A) \limplies (A \wedge \x(B \vee \f(\x B \wedge \history A)))$ instead of $B \vee \f(\x B \wedge \history A) \limplies B \vee (A \wedge \x(B \vee \f(\x B \wedge \history A)))$.

Left-to-right direction:
  \begin{displaymath}
\vcenter{
\scriptsize
\infer[\limpliesI^1]{ b: \f(\x B \wedge \history A) \limplies (A \wedge \x(B \vee \f(\x B \wedge \history A)))}{\infer[\fE^2]{ b:A \wedge \x(B \vee \f(\x B \wedge \history A))}{[ b:\f(\x B \wedge \history A)]^1 & \infer[\wedgeI]{ b:A \wedge \x(B \vee \f(\x B \wedge \history A))}{\infer[\last]{ b:A}{\infer[\refldesc^3]{ bb:A}{\infer[\historyE]{ bb:A}{\infer[\wedgeE]{ bc:\history A}{[ bc:\x B \wedge \history A]^2} & [b\descb b]^3 & [b \descb c]^2}}} & \deduce{ b:\x(B \vee \f(\x B \wedge \history A))}{\Pi_1}}}}
}
\end{displaymath}
where $\Pi_1$ is the following derivation:
  \begin{displaymath}
	\hspace*{-1.8cm}
\vcenter{
\scriptsize
\infer[\xI^4]{ b:\x(B \vee \f(\x B \wedge \history A))}
	{\infer[\splitdesc^5]{ bb':B\vee \f(\x B \wedge \history A)}
	{[b\descb c]^2 & [b\nextb b']^4 & \infer[\veeI]{ bb':B \vee \f(\x B \wedge \history A)}{\infer[\last]{ bb':B}{\infer[\xE]{ bcb':B}{\infer[\wedgeE]{ bc:\x B}{[ bc:\x B \wedge \history A]^2} & [c \descb b']^5}}} & \infer[\lin^6]{ bb':B \vee \f(\x B \wedge \history A)}{b\nextb b' & b \nextb b'' & \deduce{ bb'':B \vee \f(\x B \wedge \history A)}{\Pi_2} & [ bb':B \vee \f(\x B \wedge \history A)]^6}}
	}
}
\end{displaymath}
and $\Pi_2$ is the following derivation:
\begin{displaymath}
\hspace*{-2cm}
\vcenter{
\scriptsize
\infer[\veeI]{ bb'':B \vee \f(\x B \wedge \history A)}{\infer[\fI]{ bb'':\f(\x B \wedge \history A)}{\infer[\wedgeI]{ bb''c:\x B \wedge \history A}{\infer[\xI^7]{ bb''c:\x B}{\infer[\last]{ bb''cc':B}{\infer[\xE]{ bcc':B}{\infer[\wedgeE]{ bc:\x B}{[ bc:\x B \wedge \history A]^2} & [c \nextb c']^7}}} & \infer[\historyI^8]{ bb''c:\history A}{\infer[\last]{ bb''d:A}{\infer[\base^9]{ bd:A}{[b\nextb b'']^5 & \infer[\transdesc^{10}]{ bd:A}{[b\descb b'']^9 & [b'' \descb d]^8 & \infer[\historyE]{ bd:A}{\infer[\wedgeE]{ bc:\history A}{[ bc:\x B \wedge \history A]^2} & [b\descb d]^{10} & [d\descb c]^8}}}}}} & [b''\descb c]^5}}
}
\end{displaymath}

Right-to-left direction:
in the following derivations, we denote with $\varphi$ the formula $ b:A \wedge \x(B \vee \f(\x B \wedge \history A))$.

\begin{displaymath}
\vcenter{
\scriptsize
\infer[\limpliesI^1]{ b: (A \wedge \x (B \vee \f(\x B \wedge \history A))) \limplies \f(\x B \wedge \history A)}{\infer[\ser^2]{ b:\f (\x B \wedge \history A)}{\infer[\veeE^3]{ b:\f(\x B \wedge \history A)}{\infer[\xE]{ b e:B \vee \f(\x B \wedge \history A)}{\infer[\wedgeE]{ b:\x(B \vee \f(\x B \wedge \history A))}{[\varphi]^1} & [b \nextb e]^2} & \deduce{ b:\f(\x B \wedge \history A)}{\deduce{\Pi_1}{[ b e: B]^3}} & \deduce{ b:\f(\x B \wedge \history A)}{\deduce{\Pi_2}{[ b e:\f(\x B \wedge \history A)]^3}}}}}
}
\end{displaymath}
where $\Pi_1$ is the following derivation:
  \begin{displaymath}
\vcenter{
\scriptsize
\infer[\refldesc^4]{ b:\f(\x B \wedge \history A)}{\infer[\fI]{ b:\f (\x B \wedge \history A)}{\infer[\wedgeI]{ b b:\x B \wedge \history A}{\infer[\xI^5]{ b b:\x B}{\infer[\last]{ b b f: B}{\infer[\lin^6]{ b f:B}{[b\nextb e]^2 & [b\nextb f]^5 & [ b e:B]^3 & [ b f:B]^6}}} & \infer[\historyI^7]{ b b:\history A}{\infer[\last]{ b b':A}{\infer[\eqdesc]{ b':A}{[b\descb b']^7 & [b'\descb b]^7 & \infer[\wedgeE]{ b:A}{[\varphi]^1}}}}} & [b\descb b]^4}}
}
\end{displaymath}
$\Pi_2$ is the following derivation:
\begin{displaymath}
\hspace*{-1.7cm}
\vcenter{
\scriptsize
  \infer[\fE^8]
  { b:\f(\x B \wedge \history A)  }
  {[ b e:\f(\x B \wedge \history A)]^3 & 
    \infer[\base^9]
    { b:\f(\x B \wedge \history A)}
    {[b \nextb e]^2 & 
      \infer[\transdesc^{10}]
      { b:\f(\x B \wedge \history A)}
      {[b\descb e]^9 & [e\descb c]^8 & 
        \infer[\fI]
        { b:\f(\x B \wedge \history A)}
        {\infer[\wedgeI]
         { b c:\x B \wedge \history A}
         {\infer[\xI^{11}]{ b c:\x B}{\infer[\last]{ b c f:B}{\infer[\xE]{ b e c f:B}{\infer[\wedgeE]{ b e c:B}{[ b e c:\x B \wedge \history A]^8} & [c\nextb f]^{11}}}} 
          & 
          \deduce{ b c:\history A}{\Pi_3}
         } 
         & [b\descb c]^{10}
        }
      }
    }
  }
}
\end{displaymath}
and $\Pi_3$ is the following derivation:
\begin{displaymath}
\hspace*{-1.7cm}
\vcenter{
\scriptsize
\infer[\historyI^{12}]{ b c:\history A}{\infer[\splitdesc^{13}]{ b d:A}{[b\descb d]^{12} & \infer[\wedgeE]{ b:A}{[\varphi]^1} & \infer[\last]{ b d:A}{[ d:A]^{13}} & \infer[\last]{ b d:A}{\infer[\lin^{14}]{ b e d:A}{[b\nextb f]^{13} & [b\nextb e]^2 & [f\descb d]^{13} & \infer[\historyE]{ b e d:A}{\infer[\wedgeE]{ bec:\history A}{[ bec:\x B \wedge \history A]^8} & [e \descb d]^{14} & [d \descb c]^{12}}}}}}
}
\end{displaymath}

\begin{sidewaysfigure}[htbp]
Proof of the axiom $\mathit{(A8)}$
\scriptsize
\begin{displaymath}
\vcenter{
\scriptsize
\infer[\limpliesI^1]{ b: B \vee (\f(\x B \wedge \history A)) \limplies \f B}
 	{
		\infer[\veeE^2]{ b:\f B}{[ b:B \vee (\f (\x B \wedge \history A))]^1 
			&  
		\infer[\refldesc^3]{ b:\f B}{\infer[\fI]{ b:\f B}{\infer[\last]{ bb:B}{[ b:B]^2} & [b \descb b]^3}}
			&
		\infer[\fE^4]{ b:\f B}{[ b:\f(\x B \wedge \history A)]^2 & \infer[\ser^5]{ b:\f B}{\infer[\base^6]{ b:\f B}{[c \nextb d]^5 & \infer[\transdesc^7]{ b:\f B}{[b \descb c]^4 & [c \descb d]^6 & \infer[\fI]{ b:\f B}{\infer[\last]{ bd:B}{\infer[\xE]{ bcd:B}{\infer[\wedgeE]{ bc:\x B}{[ bc:\x B \wedge \history A]^4} & [c \nextb d]^5}} & [b\descb d]^7}}}}}
	}
	}
}
\end{displaymath}
   \label{fig:axiom-A8}
\end{sidewaysfigure}
\end{LONG}

\end{document}